\documentclass{article}
\usepackage[utf8]{inputenc}
\usepackage{amsthm}
\usepackage{graphicx,ae,aecompl}
\usepackage{amssymb}
\usepackage{amsmath}
\usepackage{subfigure}
\usepackage{undertilde}
\usepackage{algorithmic}
\usepackage[boxruled,linesnumbered]{algorithm2e}
\usepackage[left=2.5cm,top=2.5cm,right=2.5cm,bottom=2.5cm]{geometry}
\usepackage{makeidx}
\usepackage{float}
\usepackage{scalefnt}
\newtheorem{prop}{Proposition}

\title{Kurtosis control in wavelet shrinkage with generalized secant hyperbolic prior}
\author{Alex Rodrigo dos Santos Sousa \\ University of São Paulo, Brazil}

\begin{document}

\numberwithin{equation}{section}
\numberwithin{table}{section}
\numberwithin{figure}{section}

 \maketitle
    \begin{abstract}
        
The present paper proposes a bayesian approach for wavelet shrinkage  with the use of a shrinkage prior based on the generalized secant hyperbolic distribution symmetric around zero in a nonparemetric regression problem. This shrinkage prior allows the control of the kurtosis of the coefficients, which impacts on the level of shrinkage on its extreme values. Statistical properties such as bias, variance, classical and bayesian risks of the rule are analyzed and performances of the proposed rule are obtained in simulations studies involving the Donoho-Johnstone test functions. Application of the proposed shrinker in denoising Brazilian stock market dataset is also provided.  
    \end{abstract}

\section{Introduction}

Wavelet shrinkers has been succesfully applied as estimators to wavelet coefficients in denoising data problems, such as nonparametric curve estimation, time series analysis, image processing, among others, see Donoho e Johnstone (1994, 1995), Donoho et al. (1995, 1996), Vidakovic (1998, 1999) for more details about wavelet shrinkage methods. In a bayesian setup, the assumption of a shrinkage prior distribution allows us to incorporate prior information about the wavelet coefficients of the wavelet transformed unknown function, such as sparsity, support (if they are bounded), local features, and others. Several priors have been proposed along the years, see Angelini et al. (2004), Remenyi et al. (2015), Karagiannis et al. (2015), Bhattacharya et al. (2015), Torkamani et al. (2017), Griffin et al. (2017), among others.

Although the already proposed priors are well succeeded to describe important features of the wavelet coefficients, as mentioned above, there is no available prior that allows us to incorporate prior kurtosis information of the coefficients. Kurtosis control can be interesting to model heavy tailed distributed coefficients and to perform wavelet shrinkage under this context. This work addresses to this problem, with the use of a particular prior distribution that allows kurtosis control of the coefficients by convenient choice of its hyperparameters.  

As a motivational example, we generated 512 wavelet coefficients from a heavy tailed distribution and show them by the histogram in Figure \ref{fig:hist} (a). Note that most of the coefficients are zero or near from zero, once the wavelet coefficients vector of a function under certain conditions is typically sparse but there are some coefficients with high magnitudes, one of them greater than 1500 for instance. In real datasets, one observes noisy versions of wavelet coefficients, called empirical wavelet coefficients. Figure \ref{fig:hist} (b) shows the histogram of associated empirical coefficients, which are the wavelet coefficients with normal noises added. The goals of a wavelet shrinker are to remove noise effect present on empirical coefficients and estimate the wavelet coefficients of Figure \ref{fig:hist} (a). Once the wavelet coefficients are heavy tailed distributed, it is necessary a bayesian wavelet shrinker that takes the kurtosis of the distribution into account in the shrinkage process to have a good estimation performance.   

\begin{figure}[H]
\centering
\subfigure[Wavelet coefficients \label{lognormal}]{
\includegraphics[scale=0.5]{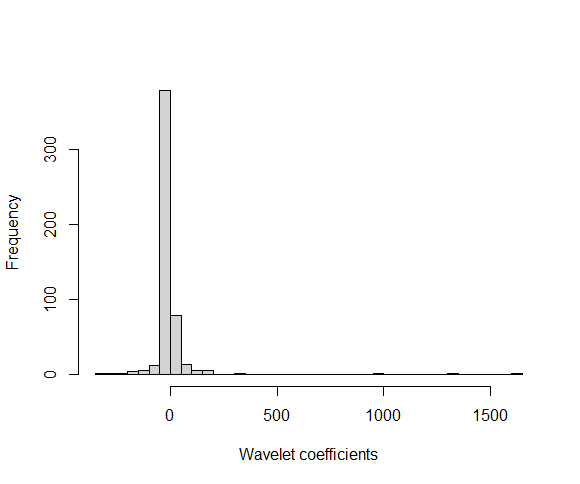}}
\subfigure[Empirical wavelet coefficients\label{blocls}]{
\includegraphics[scale=0.5]{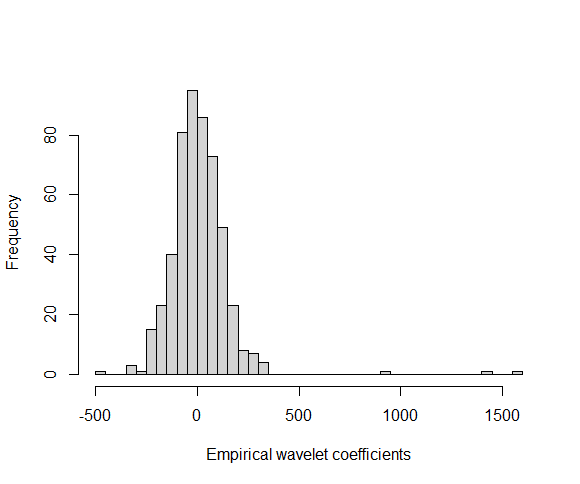}}
\caption{Histograms of generated wavelet coefficients with heavy tailed distribution and their empirical versions, with added normal noise.}\label{fig:hist}
\end{figure}

We propose the use of the generalized secant hyperbolic distribution defined by Vaughan (2002) as prior to the wavelet coefficients in a nonparametric curve estimation problem. The great advantage is the control of the kurtosis along the resolution levels as prior knowledge, which impact on the degree of shrinkage on the extremes of the empirical wavelet coefficients distribution throw convenient hyperparameters elicitation. We organize the paper defining the model in Section 2, the shrinkage rule under the generalized secant hyperbolic prior and its statistical properties are described in Section 3 and its hyperparameters elicitation discussed in Section 4. A simulation study to evaluate the performance of the proposed shrinker and compare it with standar techniques is presented in Section 5. Application of the shrinker under generalized secant hyperbolic prior in a Brazilian stock market index dataset is done in Section 6. Finally, conclusions are in Section 7.

\section{Statistical model description}

Let is consider the nonparametric regression problem of the form
\begin{equation}\label{eq:modeltime}
y_i = f(x_i) + e_i , \qquad i=1,...,n=2^J, J \in \mathbb{N},
\end{equation}
\noindent where $x_i \in [0,1]$, $i=1,...,n$, $f \in \mathbb{L}_2(\mathbb{R})= \{f:\int f^2 < \infty\}$, and $e_i$, $i=1,...,n$, are zero mean independent normal random variables with unknown variance $\sigma^2$. In vector notation, we have
\begin{equation}\label{modeltimevec}
\boldsymbol{y} = \boldsymbol{f} + \boldsymbol{e},
\end{equation}

\noindent where $\boldsymbol{y} = (y_1,...,y_n)'$, $\boldsymbol{f} = (f(x_1),...,f(x_n))'$ and $\boldsymbol{e} = (e_1,...,e_n)'$. The goal is to estimate de unknown function $f$. The standard procedure is to apply a discrete wavelet transform (DWT) on \eqref{modeltimevec}, represented by an orthogonal matrix $\boldsymbol{W}$, to obtain the following model in the wavelet domain,
\begin{equation} \label{modelvec}
\boldsymbol{d} = \boldsymbol{\theta} + \boldsymbol{\epsilon},
\end{equation}
where $\boldsymbol{d}=\boldsymbol{W}\boldsymbol{y}$, $\boldsymbol{\theta}=\boldsymbol{W}\boldsymbol{f}$ and $\boldsymbol{\epsilon}=\boldsymbol{W}\boldsymbol{e}$. For a specific coefficient of the vector $\boldsymbol{d}$, we have the simple model 
\begin{equation}\label{model}
d = \theta + \epsilon,
\end{equation}
\noindent where $d$ is the empirical wavelet coefficient, $\theta \in \mathbb{R}$ is the wavelet coefficient to be estimated and $\epsilon \sim N(0,\sigma^2)$ is the normal random error with unknown variance $\sigma^2$. Since the method works coefficient by coefficient, we extract the typical resolution level and location subindices of $d$, $\theta$ and $\epsilon$ for simplicity. Note that, according to the model \eqref{model}, $d|\theta \sim N(\theta,\sigma^2)$ and then, the problem of estimating a function $f$ becomes a normal location parameter estimation problem in the wavelet domain for each coefficient, with posterior estimation of $f$ by the inverse discrete wavelet transform (IDWT), i.e,  $\boldsymbol{\hat{f}} = \boldsymbol{W^{T}}\boldsymbol{\hat{\theta}}$. 

For bayesian estimation of $\theta$, we propose the following shrinkage prior distribution for $\theta$,
\begin{equation}\label{prior}
\pi(\theta;\alpha,\tau,t) = \alpha \delta_{0}(\theta) + (1-\alpha)g(\theta;\tau,t),
\end{equation}
where $\alpha \in (0,1)$, $\delta_{0}(\theta)$ is the point mass function at zero and $g(\theta;\tau,t)$ is the generalized secant hyperbolic (GSH) distribution symmetric around zero defined by Vaughan (2002), for $\tau > 0$ and $t > -\pi$ as
\begin{equation} \label{eq:den}
g(\theta;\tau,t) = \frac{c_1}{\tau}.\frac{\exp\{c_2.\frac{\theta}{\tau}\}}{\exp\{2c_2.\frac{\theta}{\tau}\}+2a\exp\{c_2.\frac{\theta}{\tau}\}+1}\mathbb{I}_{\mathbb{R}}(\theta),
\end{equation} 
where for $-\pi < t < 0$,
$$ a = cos(t), c_2 = \sqrt{(\pi^2-t^2)/3}, c_1 = \frac{sin(t)}{t}.c_2$$
and for $t>0$,
$$ a = cosh(t), c_2 = \sqrt{(\pi^2+t^2)/3}, c_1 = \frac{sinh(t)}{t}.c_2,$$

\noindent where $\mathbb{I}_{\mathbb{R}}(\cdot)$ is the indicator function. Note that for $t=-\pi/2$, the GSH is equivalent to the usual secant hyperbolic distribution, for $t \rightarrow 0$, GSH tends to the logistic distribution and when $t \rightarrow \infty$, GSH tends to the uniform distribution. In fact, the hyperparameter $t$ controls the kurtosis of the distribution. Vaughan shows that if $\beta = E(\theta^4)$ is the usual coefficient of kurtosis of the GSH with zero mean and $\tau=1$, then $\beta = (21\pi^2 - 9t^2)/(5\pi^2-5t^2)$ for $-\pi<t<0$ and $\beta=(21\pi^2+9t^2)/(5\pi^2+5t^2)$ for $t>0$. As t increases, the coefficient of kurtosis decreases. This allows us to control the level of shrinkage of the wavelet coefficients on the extremes values of the empirical wavelet coefficients $d$. The hyperparameters $\alpha$ and $\tau$ control the shrinkage level for empirical wavelet coefficients sufficiently close to zero. Figure \ref{fig:GSH} (a) and (b) shows some densities of the GSH distribution defined in \eqref{eq:den} and the coefficient of kurtosis as function of the parameter $t$ respectively.

\begin{figure}[H]
\centering
\subfigure[Generalized secant hyperbolic (GSH) densities.]{
\includegraphics[scale=0.4]{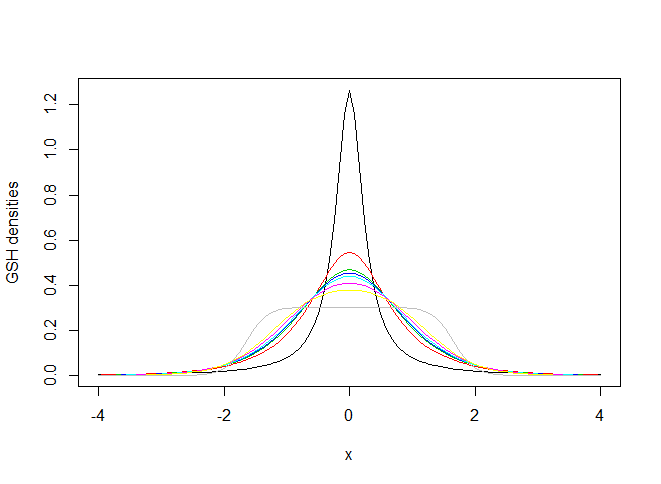}}
\subfigure[Coefficient of kurtosis.\label{kurt}]{
\includegraphics[scale=0.4]{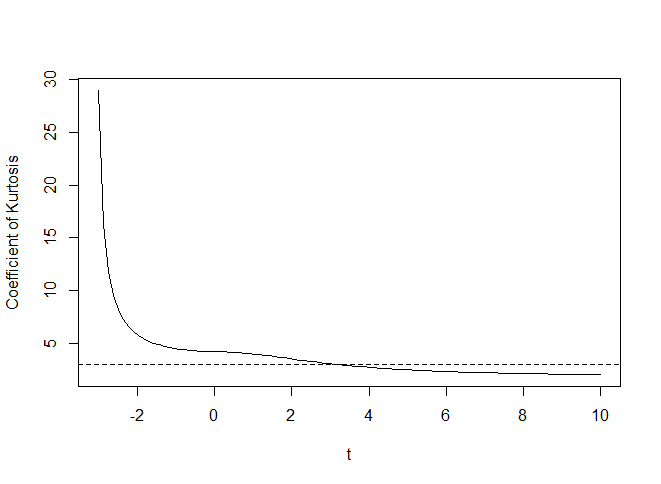}}
\caption{GSH densities and the coefficient of kurtosis as function of the GSH parameter $t$ (dashed line at coefficient of kurtosis equal to 3 - mesokurtic case).}\label{fig:GSH}
\end{figure}

\section{The shrinkage rule and its properties }
To obtain the bayesian shrinkage rule under the prior \eqref{prior}, $\delta(\cdot)$, we assume the quadratic loss function $L(\delta,\theta) = (\delta - \theta)^2$. It is well known that, under this loss function, the Bayes rule is the posterior expected value of $\theta$, i.e, $\delta(d) = E_{\pi}(\theta \mid d)$. Simple calculation proves the following proposition.

\begin{prop} 
If the prior distribution of $\theta$ is of the form $\pi(\theta;\alpha,\tau) = \alpha \delta_{0}(\theta) + (1-\alpha)g(\theta;\tau)$, where $g$ is a density function with support in $\mathbb{R}$, then the shrinkage rule under the quadratic loss function and model (2.4) is 
\begin{equation}
\delta(d) = \frac{(1-\alpha)\int_\mathbb{R}(\sigma u + d)g(\sigma u +d ; \tau)\phi(u)du}{\frac{\alpha}{\sigma}\phi(\frac{d}{\sigma})+(1-\alpha)\int_\mathbb{R}g(\sigma u +d ; \tau)\phi(u)du},
\end{equation}
where $\phi(\cdot)$ is the standard normal density function. 
\end{prop}

\begin{proof}
If $\mathcal{L}(\cdot \mid \theta)$ is the likelihood function and substituting $u=(\theta-d)/\sigma$, we have 
\begin{align*}
\delta(d) &= E_{\pi}(\theta \mid d) \\ 
          &=\frac{\int_{\mathbb{R}}\theta[\alpha\delta_{0}(\theta)+(1-\alpha)g(\theta;\tau)]\mathcal{L}(d \mid \theta)d\theta}{\int_{\mathbb{R}}[\alpha\delta_{0}(\theta)+(1-\alpha)g(\theta;\tau)]\mathcal{L}(d \mid \theta)d\theta} \\
          &= \frac{(1-\alpha)\int_{\mathbb{R}}\theta g(\theta;\tau)\frac{1}{\sqrt{2\pi}}\exp\{-\frac{1}{2}(\frac{d-\theta}{\sigma})^2\}\frac{d\theta}{\sigma}}{\alpha \frac{1}{\sigma\sqrt{2\pi}}\exp\{-\frac{1}{2}(\frac{d}{\sigma})^2\}+(1-\alpha)\int_{\mathbb{R}}g(\theta;\tau)\frac{1}{\sqrt{2\pi}}\exp\{-\frac{1}{2}(\frac{d-\theta}{\sigma})^2\}\frac{d\theta}{\sigma}}\\
          &= \frac{(1-\alpha)\int_{\mathbb{R}}(\sigma u + d)g(\sigma u + d;\tau)\phi(u)du}{\alpha \frac{1}{\sigma}\phi(\frac{d}{\sigma})+(1-\alpha)\int_{\mathbb{R}}g(\sigma u + d;\tau)\phi(u)du}.\\
\end{align*}
\end{proof}

Figure \ref{fig:rules} (a) shows the shrinkage rule under GSH prior \eqref{prior} and \eqref{eq:den}  obtained numerically for some values of $t$ in the interval $(-3,10)$, $\tau=1$ and $\alpha = 0.9$. As the hyperparameter $t$ increases, the rule $\delta$ shrinks more the coefficients on the extremes, once the kurtosis of the GSH decreases. In the curve estimation point of view, the impact of higher level of shrinkage on the extreme empirical wavelet coefficients occurs in the local features of the curve, such peaks estimation for example. In general, one can see the $\alpha$ and $\tau$ hyperparameters as usual controllers of the regularity of the curve and $t$ hyperparameter as local feature controller of the curve to be estimated, which is a novelty in terms of bayesian wavelet shrinkage.   

Figures \ref{fig:rules} (b) and \ref{fig:var} show the squared biases, variances and classical risks of the shrinkage rules considered in Figure \ref{fig:rules} (a). The properties have the same graphical symmetric behaviour, i.e, they are close to zero for $\theta$ close to zero, then there are peaks at the points when the shrinkage rule stops to shrink close to zero the empirical coefficients (in the plot, when $|\theta|$ is between 3 and 4) and finally they decrease and stabilize.

\begin{figure}[H]
\centering
\subfigure[Shrinkage rules\label{lognormal}]{
\includegraphics[scale=0.4]{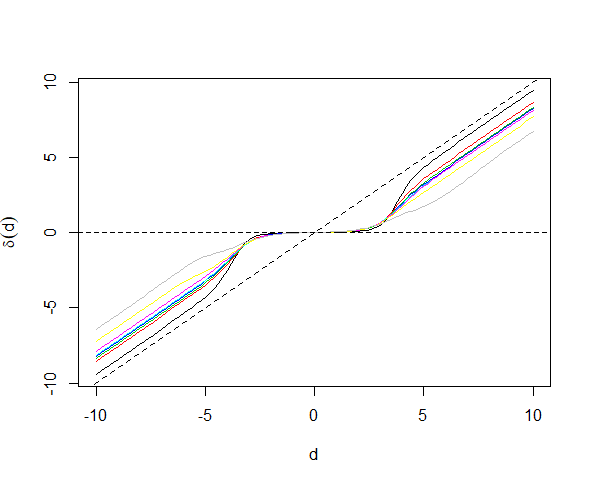}}
\subfigure[Squared Bias\label{blocls}]{
\includegraphics[scale=0.4]{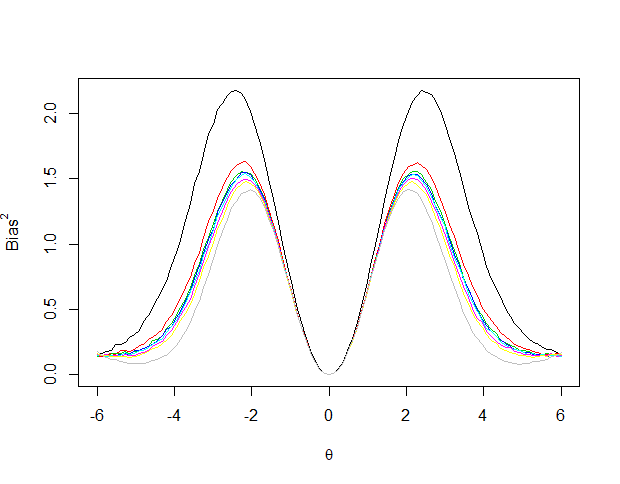}}
\caption{Shrinkage rules and their squared biases under GSH prior for some values of $t$, $\tau=1$  and $\alpha = 0.9$.}\label{fig:rules}
\end{figure}

\begin{figure}[H]
\centering
\subfigure[Variances\label{lognormal}]{
\includegraphics[scale=0.4]{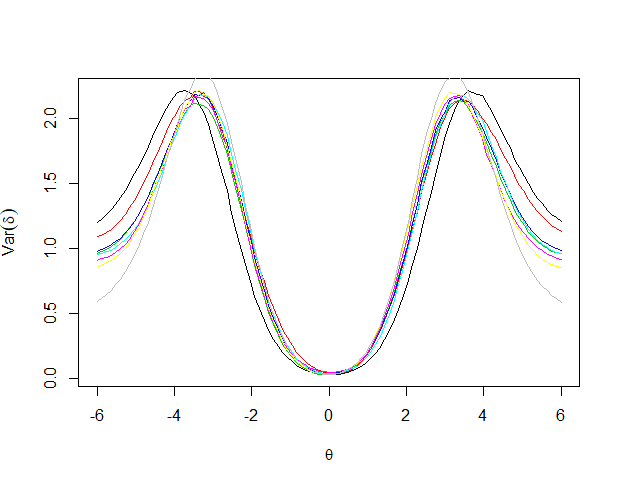}}
\subfigure[Classical Risks\label{blocls}]{
\includegraphics[scale=0.4]{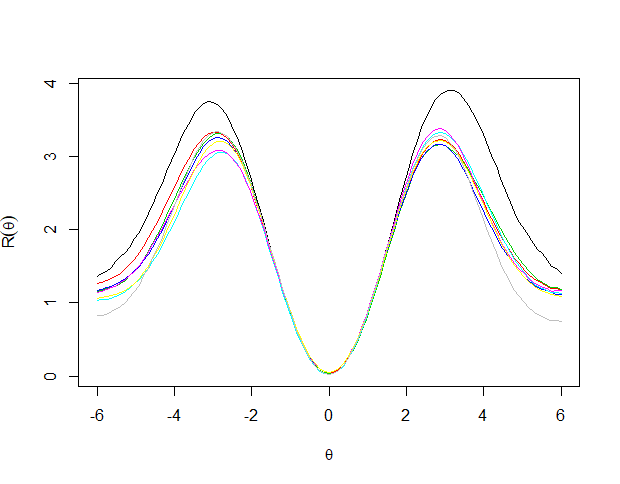}}
\caption{Variances and classical risks $R(\theta)$ of the shrinkage rule under GSH prior for some values of $t$, $\tau=1$  and $\alpha = 0.9$.}\label{fig:var}
\end{figure}

Tables \ref{tab:briskt} and \ref{tab:briska} show the bayes risks $r_\delta$ of the shrinkage rules for some values of $t$ and $\alpha$ respectively. In both cases, there is not a monotonical relationship observed between the risk and the considered hyperparameters but the risks themselves are close to zero for the hyperparameters values considered.

\begin{table}[H]
\centering
\label{my-label}
\begin{tabular}{lllllllll}
\hline
$t$ & -3   & -2  & -1  & 0.1 & 1 & 2 & 3 & 10   \\ \hline
$r_\delta$ & 0.125 & 0.329 & 0.223 & 0.235 & 0.183 & 0.221 & 0.240 & 0.247 \\ \hline
\end{tabular}
\caption{Bayes risks of the shrinkage rules under GSH prior for some values of $t$, $\alpha = 0.9$ and $\tau = 1$.} \label{tab:briskt}
\end{table}

\begin{table}[H]
\centering
\label{my-label}
\begin{tabular}{llllll}
\hline
$\alpha$ & 0.6   & 0.7  &  0.8  & 0.9 & 0.99    \\ \hline
$r_\delta$ & 0.744 & 0.762 & 0.129 & 0.245 & 0.03  \\ \hline
\end{tabular}
\caption{Bayes risks of the shrinkage rules under GSH prior for some values of $\alpha$ , $t = 3$ and $\tau=1$.}\label{tab:briska}
\end{table}

\section{Parameters elicitation}
The elicitations of the hyperparameters $\alpha$, $t$ and $\tau$ and the parameter $\sigma$ are crucial to our bayesian procedure be well succeed. In fact, there already are propositions well accepted in the area for values of $\sigma$ and $\alpha$. Based on the fact that much of the noise information present in the data can be obtained on the finest resolution scale, for the robust $\sigma$ estimation, Donoho and Johnstone (1994) suggested the robust estimator
\begin{equation}\label{eq:sigma}
\hat{\sigma} = \frac{\mbox{median}\{|d_{J-1,k}|:k=0,...,2^{J-1}\}}{0.6745}.
\end{equation}
Angelini et al. (2004) suggested the hyperparameter $\alpha$ be dependent on the resolution level $j$ according to the expression
\begin{equation}\label{eq:alpha}
\alpha = \alpha(j) = 1 - \frac{1}{(j-J_{0}+1)^\gamma},
\end{equation}
where $J_0$ is the primary resolution level and $\gamma > 0$. They also suggested that in the absence of additional information, $\gamma = 2$ can be adopted. 

For the hyperparameter $t$, we suggest using its relationship with the coefficient of kurtosis obtained by Vaughan, but adopting the sample coefficient of kurtosis of the empirical wavelet coefficients as estimator of the kurtosis, i.e,
\begin{equation}\label{eq:t}
t=
\begin{cases}
-\pi\sqrt{\frac{5\hat{\beta}-21}{5\hat{\beta}-9}},&\mbox{if}\quad \hat{\beta}\geq 4.2,\\
\pi\sqrt{\frac{21-5\hat{\beta}}{5\hat{\beta}-9}}, &\mbox{if}\quad \hat{\beta}<4.2,
\end{cases}
\end{equation} 

\noindent where $\hat{\beta}$ is the sample coefficient of kurtosis of the empirical wavelet coefficients, i.e, $\hat{\beta} = (1/n)\sum_{i} (d_i - \bar{d})^{4} / [(1/n)\sum_{i} (d_i - \bar{d})^{2}]^2$.

The hyperparameter $\tau$ has the same impact as the hyperparameter $\alpha$ in the shrinkage procedure, since it controls the shape of the GSH distribution around zero in our context. In this sense, we consider $\tau=1$ and control the level of shrinkage for empirical coefficients around zero by $\alpha$ elicitation in our study.

\section{Simulation studies}

Simulation studies were done to evaluate the performance in terms of averaged mean square error (AMSE) of the shrinkage rule under GSH prior and compare it with well known methods in the literature and practice. To this goal, the so called Donoho and Johnstone test functions, Bumps, Blocks, Doppler and Heavisine, see Donoho and Johnstone (1994), were used to simulate data in three sample sizes, $n = 512, 1024, 2048$ and three signal to noise ratios (SNR), SNR $=3, 5,7$. The Daubechies wavelets with ten null moments  (Daub10) were  applied to the generated data.

We compare the shrinkage rule under GSH prior (GSH rule) with fixed $\tau = 1$ and $\alpha$ and $t$ according to \eqref{eq:alpha} and \eqref{eq:t} with the methods Cross Validation (CV) of Nason (1996), False Discovery Rate (FDR) of Abramovich and Benjamini (1996), Stein Unbiased Risk Estimator (SURE) of Donoho and Johnstone (1995), Large Posterior Mode (LPM) of Cutillo et al. (2008) and Amplitude-scale invariant Bayes Estimator (ABE) of Figueiredo and Nowak (2001). 

We used the mean squared error (MSE), $MSE = \frac{1}{n} \sum_{i=1}^{n}[{\hat f(x_i)} - f(x_i)]^2$ as performance measure of the shrinkage rules. For each function, the process was repeated $M = 200$ times and a comparison measure, the average of the obtained MSEs, $AMSE = \frac{1}{M} \sum_{j=1}^{M}MSE_j$, was calculated for each rule as shown in Tables \ref{tab:amse1} and \ref{tab:amse2}. The estimates of the curves by the shrinkage rules for $n=2048$ are shown in Figure \ref{fig:curves}.

In general, the performance of the shrinkage rule under GSH prior was great in the simulations, with the lower or the second lower AMSE in almost all the scenarios. The rule sometimes lost to ABE method, mainly for Bumps and Blocks functions, but for Doppler and Heavisine, our rule had the best performance in terms of AMSE. Further, the AMSE decreased as the SNR decreased and/or de sample size increased, as expected.

\begin{table}[H]
\scalefont{0.5}
\centering
\label{my-label}
\begin{tabular}{|c|c|c|c|c|c|||c|c|c|c|c|c|}
\hline
Signal & n & Method & SNR = 3 & SNR = 5 & SNR =7 & Signal & n & Method & SNR = 3 & SNR = 5 & SNR =7   \\ \hline

Bumps&512	&	Univ	&	11.114	&	5.178	&	3.027	& Blocks	&512	&	Univ	&	6.928	&	3.676	&	2.262	\\
&	&	CV	&	11.466	&	9.488	&	6.297	&	&	&	CV	&	2.555	&	1.247	&	0.843	\\
&	&	FDR	&	9.290	&	4.323	&	2.633	&	&	&	FDR	&	5.904	&	2.919	&	1.747	\\
&	&	SURE	& 3.666	&	1.562	&	0.892	&	&	&	SURE	&	2.812	&	1.224	&	0.691	\\
&	&	LPM	&	5.446	&	1.948	&	1.000	&	&	&	LPM	&	5.417	&	1.959	&	0.995	\\
&	&	ABE	&	\textbf{2.713}	&\textbf{1.055}	&\textbf{0.577}	&	&	&	ABE	&\textbf{2.289}	&\textbf{0.930}	&	\textbf{0.480}	\\
&	&	GSH rule	& \textbf{3.150} &\textbf{1.120}	&\textbf{0.598}	&	&	&	GSH rule	&	2.979	&\textbf{1.186}	&\textbf{0.567}	\\ \hline
																			
&1024	&	Univ	&	7.556	&	3.576	&	2.142	&	&1024	&	Univ	&	4.879	&	2.488	&	1.548	\\
&	&	CV	&	2.942	&	1.935	&	1.739	&	&	&	CV	&	1.794	&	0.842	&	0.537	\\
&	&	FDR	&	5.608	&	2.543	&	1.478	&	&	&	FDR	&	3.888	&	1.867	&	1.125	\\
&	&	SURE	&	2.512	&	1.060	&	0.592	&	&	&	SURE	&	1.900	&	0.843	&	0.484	\\
&	&	LPM	&	5.469	&	1.969	&	0.999	&	&	&	LPM	&	5.457	&	1.964	&	1.005	\\
&	&	ABE	&\textbf{1.879}	&\textbf{0.751}	&\textbf{0.401}	&	&	&	ABE	&\textbf{1.549}	&\textbf{0.644}	&	\textbf{0.349}	\\
&	&	GSH rule&\textbf{2.158}	&\textbf{0.814}	&\textbf{0.425}	&	&	&	GSH rule	& 1.910	&\textbf{0.760}	&	\textbf{0.407}	\\ \hline
																			
&2048	&	Univ	&	5.056	&	2.349	&	1.392	&	&2048	&	Univ	&	3.423	&	1.775	&	1.103	\\
&	&	CV	&	1.606	&	0.734	&	0.480	&	&	&	CV	&	1.304	&	0.589	&	0.352	\\
&	&	FDR	&	3.594	&	1.598	&	0.920	&	&	&	FDR	&	2.681	&	1.292	&	0.766	\\
&	&	SURE	&	1.659	&	0.699	&	0.390	&	&	&	SURE	&	1.359	&	0.602	&	0.341	\\
&	&	LPM	&	5.442	&	1.959	&	0.999	&	&	&	LPM	&	5.425	&	1.953	&	0.996	\\
&	&	ABE	& \textbf{1.253}	&\textbf{0.497}	&\textbf{0.267}	&	&	&	ABE	&\textbf{1.141}	&\textbf{0.464}	&	\textbf{0.250}	\\
&	&	GSH rule&\textbf{1.335}	&\textbf{0.508}	&\textbf{0.271}	&	&	&	GSH rule	&\textbf{1.298}	&\textbf{0.531}	&\textbf{0.278}	\\ \hline

\end{tabular}
\caption{AMSE of the shrinkage/thresholding rules in the simulation study for Bumps and Blocks DJ-test functions.}\label{tab:amse1}
\end{table} 

\begin{table}[H]
\scalefont{0.5}
\centering
\label{my-label}
\begin{tabular}{|c|c|c|c|c|c|||c|c|c|c|c|c|}
\hline
Signal & n & Method & SNR = 3 & SNR = 5 & SNR =7 & Signal & n & Method & SNR = 3 & SNR = 5 & SNR =7   \\ \hline

Doppler&512	&	Univ	&	2.639	&	1.383	&	0.885	&Heavisine	&512	&	Univ	&	0.565	&	0.403	&	0.302	\\
&	&	CV	&	1.282	&	0.632	&	0.444	&	&	&	CV	&	0.498	&	0.275	&	0.173	\\
&	&	FDR	&	2.542	&	1.247	&	0.760	&	&	&	FDR	&	0.594	&	0.436	&	0.308	\\
&	&	SURE	&	1.323	&	0.578	&	0.334	&	&	&	SURE	&	0.568	&	0.414	&	0.313	\\
&	&	LPM	&	5.472	&	1.956	&	1.005	&	&	&	LPM	&	5.418	&	1.956	&	0.995	\\
&	&	ABE	&\textbf{1.161}	&\textbf{0.470}	&\textbf{0.255}	&	&	&	ABE	&	0.702	&	0.311	&	0.177	\\
&	&	GSH rule	&\textbf{1.162}	&\textbf{0.500}	&\textbf{0.264}	&	&	&	GSH rule	&\textbf{0.472}	&\textbf{0.252}	&\textbf{0.159}	\\ \hline
																			
&1024	&	Univ	&	1.632	&	0.848	&	0.533	&	&1024	&	Univ	&	0.455	&	0.313	&	0.230	\\
&	&	CV	&	0.806	&	0.369	&	0.219	&	&	&	CV	&	0.365	&	0.198	&	0.125	\\
&	&	FDR	&	1.543	&	0.752	&	0.453	&	&	&	FDR	&	0.499	&	0.322	&	0.225	\\
&	&	SURE	&	0.836	&	0.381	&	0.223	&	&	&	SURE	&	0.459	&	0.320	&	0.238	\\
&	&	LPM	&	5.443	&	1.959	&	1.000	&	&	&	LPM	&	5.439	&	1.958	&	1.001	\\
&	&	ABE	&	0.810	&	0.333	&	0.186	&	&	&	ABE	&	0.592	&	0.245	&	0.135	\\
&	&	GSH rule	&\textbf{0.711}	&\textbf{0.295}	&\textbf{0.171}	&	&	&	GSH rule	&\textbf{0.340}	&\textbf{0.176}	&\textbf{0.106}	\\ \hline
																			
&2048	&	Univ	&	1.160	&	0.583	&	0.367	&	&2048	&	Univ	&	0.362	&	0.235	&	0.166	\\
&	&	CV	&	0.559	&	0.254	&	0.147	&	&	&	CV	&	0.268	&	0.142	&	0.088	\\
&	&	FDR	&	1.030	&	0.487	&	0.291	&	&	&	FDR	&	0.398	&	0.234	&	0.157	\\
&	&	SURE	&	0.578	&	0.259	&	0.148	&	&	&	SURE	&	0.365	&	0.239	&	0.170	\\
&	&	LPM	&	5.431	&	1.955	&	0.997	&	&	&	LPM	&	5.439	&	1.958	&	0.999	\\
&	&	ABE	&	0.634	&	0.251	&	0.133	&	&	&	ABE	&	0.515	&	0.204	&	0.110	\\
&	&	GSH rule	&\textbf{0.423}	&\textbf{0.189}	&\textbf{0.104}	&	&	&	GSH rule	&\textbf{0.238}	&\textbf{0.120}	&\textbf{0.069}	\\ \hline

\end{tabular}
\caption{AMSE of the shrinkage/thresholding rules in the simulation study for Doppler and Heavisine DJ-test functions.}\label{tab:amse2}
\end{table} 

\begin{figure}[H]
\subfigure[SNR = 3\label{lognormal}]{
\includegraphics[scale=0.4]{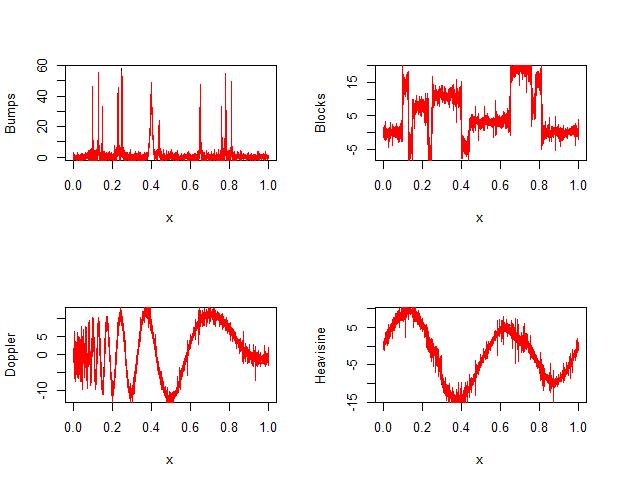}}
\subfigure[SNR = 5\label{blocls}]{
\includegraphics[scale=0.4]{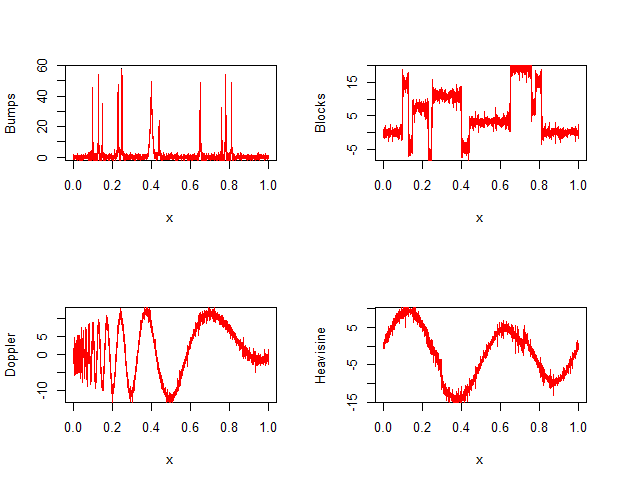}}
\subfigure[SNR = 7\label{blocls}]{
\includegraphics[scale=0.4]{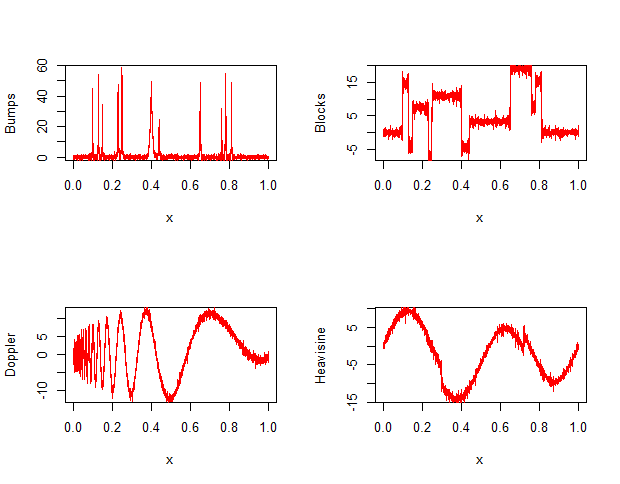}}
\caption{Curve estimates of the shrinkage rule under GSH prior with $\tau=1$ for $n=2048$.}\label{fig:curves}
\end{figure}

\section{Real data application}
To verify the performance of the GSH shrinkage rule in a real data, we consider a São Paulo stock market index (IBOVESPA) time series collected daily from January, 1st 2019 to January, 29th 2021 ($n=512$), which is showed in Figure \ref{fig:ibov} (a). After a DWT application, with a Daubechies basis with 10 null moments, we obtain the empirical wavelet coefficients, showed by resolution level in Figure \ref{fig:ibov} (b). Denoising is important in this context to detect sistematic changes of the IBOVESPA. In this sense, wavelet shrinkage is one of the most used denoising tools in time series.

Figures \ref{fig:ibov2} (a) and (b) shows denoised IBOVESPA time series after GSH shrinkage rule application and the estimated wavelet coefficients respectively. Observe that most of the empirical coefficients are shrunk to zero or near to zero and only the significative empirical coefficients, in magnitude, remain nonzero after shrinkage process. It minimizes the effect of noise and makes clearer the main features of the time series.

\begin{figure}[H]
\centering
\subfigure[IBOVESPA dataset\label{lognormal}]{
\includegraphics[scale=0.4]{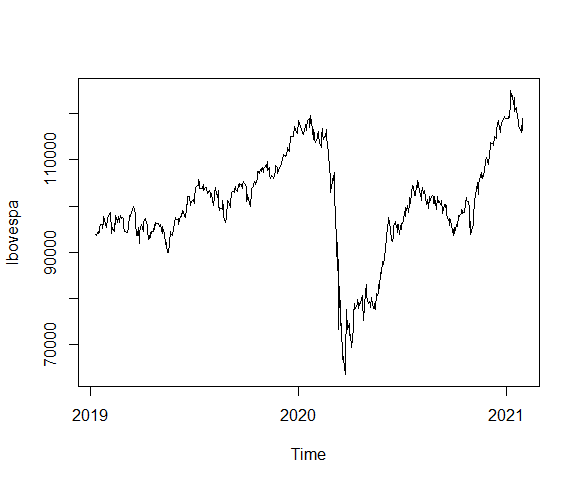}}
\subfigure[Empirical coefficients\label{blocls}]{
\includegraphics[scale=0.4]{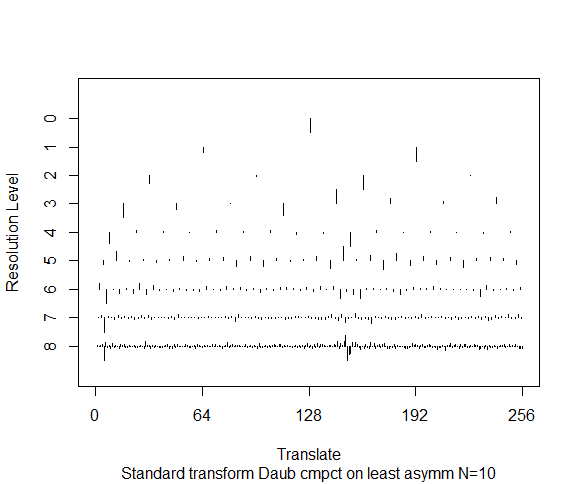}}
\caption{São Paulo stock market index (IBOVESPA) dataset from January, 1st 2019 to January, 29th 2021 (a) and the associated empirical wavelet coefficients by resolution level after a discrete wavelet trasnformation (b).}\label{fig:ibov}
\end{figure}

\begin{figure}[H]
\centering
\subfigure[Denoised data\label{lognormal}]{
\includegraphics[scale=0.4]{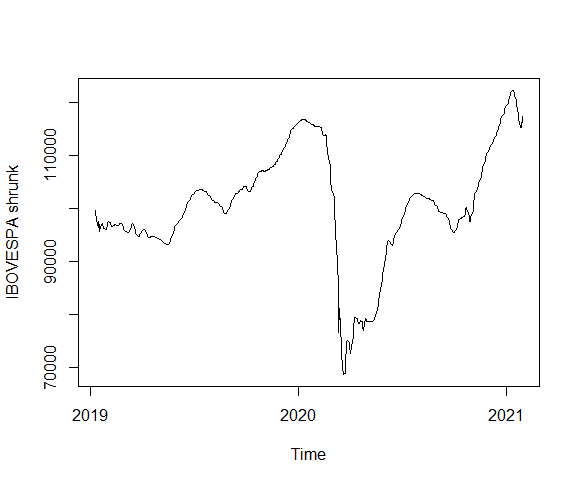}}
\subfigure[Estimated coefficients\label{blocls}]{
\includegraphics[scale=0.4]{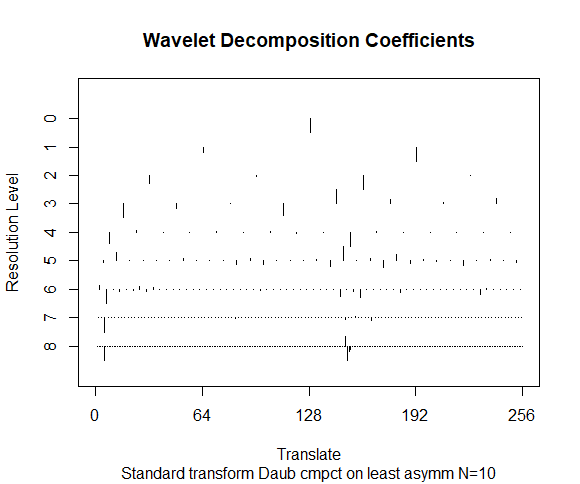}}
\caption{Denoised São Paulo stock market index (IBOVESPA) dataset from January, 1st 2019 to January, 29th 2021 after GSH shrinkage rule application (a) and the estimated wavelet coefficients by resolution level.(b).}\label{fig:ibov2}
\end{figure}

\section{Conclusions}
The simulation studies showed us that the shrinkage rule under GSH prior had a great performance in terms of averaged mean square error when compared with some of the most used wavelet shrinkage and thresholding techniques nowadays. Further, the advantage of controlling the level of shrinkage on the extreme values of the empirical coefficients through the kurtosis of the GSH distribution allows to include information about local features of the curve to be estimated such peaks, for example. These features can allow the proposed shrinkage rule be considered for practitioners for real data applications.     

\section{Acknowledgement}
This study was financed by the Coordenação de Aperfeiçoamento de Pessoal de Nível Superior – Brasil (CAPES).

\end{document}